\documentclass[11pt]{article}

\usepackage{fullpage}
\usepackage{graphicx,amsmath,amssymb,mathtools,amsthm}

\newtheorem{theorem}{Theorem}[section]

\newtheorem{lemma}[theorem]{Lemma}

\newtheorem{definition}[theorem]{Definition}
\newtheorem{observation}[theorem]{Observation}

\newtheorem{open}[theorem]{Open Problem}

\title{Some properties are not even partially testable\thanks{The research leading to these results has received funding from the European Union's - Seventh Framework Programme [FP7/2007-2013] under grant agreement n° 202405 (PROPERTY TESTING).}}
\author{Eldar Fischer\thanks{Department of Computer Science, Technion, Haifa 32000, Israel. \mbox{eldar@cs.technion.ac.il}} \and Yonatan Goldhirsh\thanks{Department of Computer Science, Technion, Haifa 32000, Israel. \mbox{jongold@cs.technion.ac.il}} \and Oded Lachish\thanks{Birkbeck, University of London, London, UK. \mbox{oded@dcs.bbk.ac.uk}}}
\begin{document}
\maketitle
\thispagestyle{empty}
\begin{abstract}
For a property $P$ and a sub-property $P'$, we say that $P$ is \emph{$P'$-partially testable with $q$ queries} if there exists an algorithm that distinguishes, with high probability, inputs in $P'$ from inputs $\epsilon$-far from $P$ by using $q$ queries. There are natural properties that require many queries to test, but can be partitioned into a small number of subsets for which they are partially testable with very few queries, and in fact the minimal $O(1/\epsilon)$.

We prove that this is not always the case. More than that, we prove the existence of a property $P$ such that the only subsets $P'$ for which $P$ is $P'$-partially testable are very small. To prove this we introduce new techniques for proving property testing lower bounds. In addition to obtaining some broad-brush criteria for non-testability, this implies a lower bound on the possibility of PCPPs with a sublinear proof size. This also implies lower bounds on MAPs, a notion newly defined by Gur and Rothblum.

The new techniques rely on analyzing a proposed partial tester. We show that the queries performed by a tester must, with high probability, query indexes where a uniformly random member of the sub-property has low entropy. We then show how one can aggregate the ``entropy loss'' to deduce that a random choice in the sub-property must have low entropy, and therefore the sub-property must be small.

We develop two techniques for aggregating the entropy loss. A simpler technique that applies to non-adaptive testers is based on partitioning the input bits into high query probability parts and parts where there is an entropy loss when conditioned on the high probability parts. Against adaptive testers we develop a more intricate technique based on constructing a decision tree. The root-to-leaf paths in this tree rearrange the input into parts where each part exhibits entropy loss when conditioned on the path prefix. This decision tree is constructed by combining carefully selected decision trees from those used by the adaptive testing algorithm.
\end{abstract}

\newpage
\pagenumbering{arabic}
\section{Introduction}

Property Testing deals with randomized approximation algorithms that operate under low information situations. Formally, we deal with objects from some universe $U$ parametrized by an integer $n$, usually $\Sigma^n$ where $\Sigma$ is some finite alphabet; with a notion of \emph{distance} between two objects in $U$, usually the Hamming distance; and with a notion of a \emph{query} to an object in $U$, usually corresponding to retrieving $x_i$ for an index $i\in \{1,\ldots,n\}$. 

\begin{definition}[Testable property]
Let $P\subseteq\{0,1\}^n$. We say that $P$ is \emph{testable with $q$ queries} if there exists an algorithm $A$ that gets as input a parameter $\epsilon>0$ and query access to an input string $x\in \{0,1\}^n$ and outputs \emph{accept} or \emph{reject} such that:
\begin{itemize}
\item If $x\in P$, then $A$ accepts with probability at least $2/3$.
\item If $d(x,P)>\epsilon$, then $A$ rejects with probability at least $2/3$.
\end{itemize}
If furthermore all queries performed to the input can be decided before any of them are made, then the algorithm is \emph{non-adaptive}, and otherwise it is \emph{adaptive}. If we require that whenever $x\in P$, then the algorithm accepts with probability $1$, then the algorithm is \emph{1-sided}, and otherwise it is \emph{2-sided}.
\end{definition}

Property Testing was first addressed by Blum, Luby and Rubinfeld~\cite{BlumLR93}, and most of its general notions were first formulated by Rubinfeld and Sudan~\cite{RubinfeldS96}. The first investigated properties were mostly of an algebraic nature, such as the property of a Boolean function being linear. The first investigation of combinatorial properties and the formal definition of testability was by Goldreich, Goldwasser and Ron~\cite{GGR}. Since then Property Testing has attracted significant attention. For surveys see \cite{Fischer01theart,Ron2008,Ron2010}.

When proving that testing a property requires many queries, one might ask ``how strong is this requirement?'', which can be illustrated with an example. Alon et. al.~\cite{AlonKNS00} studied the testability of formal languages, and proved that the language $L=\{uu^Rvv^R\vert u,v\in\{0,1\}^*\}$ requires at least $\Omega(\sqrt{n})$ queries to test. Informally, one may say that the ``reason'' for this language being untestable is the difficulty in guessing the length of $uu^R$. This can be made formal by considering the languages $L_i=\{uu^Rvv^R\vert u,v\in\{0,1\}^*, |u|=i\}$, which form a partition of $L$. A simple sampling algorithm can perform $O(\epsilon^{-1})$ queries to an input and distinguish between inputs in $L_i$ and inputs $\epsilon$-far from $L$. It is also important to note that $|L\cap\{0,1\}^n|=2^{\Theta(n)}$, but its partition $L_0\cap\{0,1\}^n,\ldots,L_n\cap\{0,1\}^n$ is only to a number of subsets linear in $n$.

This phenomenon is not unique to the language considered by Alon et. al. Another example is that of graph isomorphism, first considered in the property testing framework by Alon et. al.~\cite{AlonFKS00} (and later by Fischer and Matsliah~\cite{FischerM08}), and shown to require at least $\Omega(n)$ queries to test. In this setting we consider a pair of unknown graphs given by their adjacency matrices, and we are charged with distinguishing the case where they are isomorphic from the case where more than $\epsilon n^2$ of their edges must be changed to make them isomorphic. In this case, the size of the property is $2^{\Theta(n^2)}$, and we can partition the property into $n!$ properties $\{P_{\pi}\vert \pi\in S_n\}$, each defined by $P_{\pi}=\{(G_1,G_2)\vert \pi(G_1)=G_2\}$, such that a sampling algorithm can perform $O(\epsilon^{-1})$ queries to an input and distinguish between inputs in $P_{\pi}$ and inputs $\epsilon$-far from the original property.

Thus it is tempting to ask whether this is a general phenomenon. Can any property $P$ be partitioned into $k=|P|^{o(1)}$ properties $P_1,\ldots,P_k$ such that the task of distinguishing inputs in $P_i$ from inputs far from $P$ can be performed with a number of queries that depends only on $\epsilon$? The main result of the present paper is to prove that this is not always the case. In fact, there exist properties for which any such partition must be to a number of subsets exponential in $n$.

To prove this result we consider a simpler setting that implies it. 
\begin{definition}[Partially testable property]
Let $P\subseteq\{0,1\}^n$ and $P'\subseteq P$. We say that $P$ is \emph{$P'$-partially testable with $q$ queries} if there exists an algorithm $A$ that gets as input a parameter $\epsilon>0$ and query access to an input string $x\in \{0,1\}^n$ and outputs \emph{accept} or \emph{reject} such that:
\begin{itemize}
\item If $x\in P'$, then $A$ accepts with probability at least $2/3$.
\item If $d(x,P)>\epsilon$, then $A$ rejects with probability at least $2/3$.
\end{itemize}
If furthermore all queries performed to the input can be decided before any of them are made, then the algorithm is \emph{non-adaptive}, and else it is \emph{adaptive}.
\end{definition}

Obviously, if $P$ is testable with $q$ queries, then for any subset $P'\subseteq P$ it is $P'$-partially testable with the same number of queries. On the other hand, for any property $P$ and any element $x\in P$, we have that $P$ is $\{x\}$-partially testable with $O(\epsilon^{-1})$ queries.

The partitions described above are in fact partitions of $P$ into subsets $P_1,\ldots,P_k$ such that $P$ is $P_i$-partially testable for every $1\leq i\leq k$. If there exists such a partition into not too many sets, then there must be at least one set that is relatively large. Our main result shows that there exists a property $P$ for which all subsets $P'\subseteq P$ such that $P$ is $P'$-partially testable are small. In fact, all linear codes with large dual distance define such properties.

\begin{theorem}\label{thm:mainthm}
Let $C\subseteq \{0,1\}^n$ be a linear code of size $|C|\leq 2^{\frac{1}{64}n}$ and dual distance $\Gamma$. For every $C'\subseteq C$, if $C$ is $C'$-partially testable with $q$ adaptive queries, then $|C'|\leq |C|2^{-\Theta\left(\frac{\Gamma}{q}\right)}$.
\end{theorem}

We will first prove a weak version of Theorem \ref{thm:mainthm} in Section~\ref{sec:nonadap} which will apply for $q$ non-adaptive queries and imply the bound $|C'|\leq |C|2^{-\Theta\left(\frac{\Gamma}{q^3}\right)}$. This proof will use some of the key ideas that will later manifest in the proof of the theorem in its full generality in Section~\ref{sec:adap}.

An important question is the existence of codes with strong parameters. A random linear code $C$ will have $\Gamma=\Theta(n)$ and $|C|=2^{\Theta(n)}$ with high probability (this is implied by the Gilbert-Varshamov bound~\cite{Gilbert,Varshamov}; MacWilliams et. al.~\cite{MacWST72} showed that this can also be obtained by codes which are self-dual and thus also have good distance), and thus by Theorem \ref{thm:mainthm} we will have that for any $C'\subseteq C$ such that $C$ is $C'$-partially testable with $q$ queries, $|C'|\leq |C|2^{-\Theta\left(\frac{n}{q}\right)}$. For a constant $q$, this implies that partial testability will only be possible with exponentially small subsets. The best explicit (and low uniform decision complexity) construction known to us is that of \cite{DWiseABI}, which gives $|C|=2^{\Theta(n)}$ with $\Gamma=\Theta(n/\log n)$, and thus the bound becomes $|C'|\leq n^{O(1)}|C|2^{-\Theta\left(\frac{n}{q}\right)}$, which is polynomially worse than the non-explicit bound, but is still a strong upper bound on the size of $C'$.

Theorem~\ref{thm:mainthm} implies that there exist properties $P$ that require a lot of queries to test, and that every partition of $P$ into subsets $P_1,\ldots,P_k$ such that $P$ is $P_i$-partially testable for every $1\leq i\leq k$ requires that $k$ will be very big. One might ask if we can prove a converse. That is, if $P$ can be tested with a few queries, can we find such a partition with a small $k$?

\begin{open}
Let $P$ be a property testable with $r$ queries. Is it true that we can partition $P$ into subsets $P_1,\ldots, P_k$ such that $P$ is $P_i$-partially testable with $O(1)$ queries for every $1\leq i\leq k$ and $k$ is bounded by some moderate function of $r$?
\end{open}

Theorem~\ref{thm:mainthm} implies that for some properties, $k$ might be as big as $2^{\Theta\left(\frac{n}{q}\right)}$. It is not clear whether this value of $k$ can always be obtained. The trivial upper bound for every property is by partitioning into $2^{n-q}$ subsets of size $2^q$. Are there properties for which this is required?

\begin{open}
Does there exist a property $P$ such that for every $P'\subseteq P$ where $P$ is $P'$-partially testable with $q$ queries we also have $|P'|\leq |P| 2^{\Theta(q)-\Theta(n)}$?
\end{open}

\subsection{Related work}

Our notion of a partition is similar to existing notions in computational complexity. For  a partition $P=P_1\cup P_2\cup\ldots\cup P_k$ where for every $1\leq i\leq k$, $P$ is $P_i$-partially testable, the designation of $P_i$ can be seen as a ``proof'' that a certain $x$ is in $P$. If $x\in P$, then there exists some $P_i$ such that $x\in P_i$ and therefore a $P_i$-partial tester for $P$ will accept it with high probability. If $x$ is $\epsilon$-far from $P$, then all $P_i$-partial testers for $P$ will reject it with high probability. 

This is similar to the notion of a \emph{Probabilistically Checkable Proof of Proximity (PCPP)}, first introduced by Ben-Sasson et. al.~\cite{PCPPs}.  PCPPs are to property testing as NP is to P. A $q$ query PCPP for a property $P\subset U$ is an algorithm that gets as input $x\in U$ and a \emph{proof of proximity} $\pi\in \{0,1\}^l$. The algorithm must perform at most $q$ queries to $x$ and $\pi$ and fulfill the requirement that if $x\in P$ then there exists a proof $\pi$ that causes the algorithm to accept with high probability, but when $x$ is $\epsilon$-far from $P$ then for any proof $\pi$ the algorithm rejects with high probability. In our setting, the algorithm is allowed free access to a proof of length $l=\log(k)$, but we expect $l$ to be sublinear in the size of $x$. In particular, the property we analyze here cannot have a PCPP with a sublinear length proof.

Rothblum et. al. ~\cite{IPP} introduced the notion of an \emph{Interactive Proof of Proximity (IPP)}. In an IPP for a property $P$, the tester can also communicate with a \emph{prover} in addition to querying the input $x$. If $x\in P$ then the prover has a strategy that will cause the tester to accept with high probability. When $x$ is $\epsilon$-far from $P$, the prover cannot make the tester accept with high probability. Rothblum et. al. show that all languages in NC admit such a protocol with $\sqrt{n}$ query and communication complexity and $\mathrm{polylog}(n)$ communication rounds. Protocols of this kind are only interesting for the case where the communication complexity is sublinear, or else the prover may just give the input to the tester.

Independently of the present work, Gur and Rothblum~\cite{MAP} weakened the IPP model to create \emph{Merlin-Arthur Proofs of Proximity (MAP)}. Gur and Rothblum define a MAP as a proof-system for a property $P$ where for an input $x$ and a proof $\pi$ the verifier reads the entire proof $\pi$ and queries $q$ bits from the input $x$. If $x\in P$, then there exists a proof $\pi$ such that the verifier accepts with high probability, and if $x$ is far from $P$, then for every proof $\pi$ the verifier rejects with high probability. Since we can trivially set $\pi=x$, the only interesting cases are where the length of $\pi$ is sublinear.

The notion of a MAP with $q$ queries and proofs of length $\ell$ for a property $P$ is equivalent to the existence of $k=2^{\ell}$ sets $P_1,\ldots,P_{k}$ such that $P=P_1\cup P_2\cup\ldots\cup P_{k}$ where for every $1\leq i\leq k$, $P$ is $P_i$-partially testable with $q$ queries. 

Gur and Rothblum give several constructions of properties where a MAP with a sublinear length proof greatly reduces query complexity. Gur and Rothblum also introduce the \emph{Tensor Sum} family of properties, for which they prove that for every constant $\alpha>0$ there exists an instantiation of Tensor Sum such that any MAP for it that performs $q$ queries must require a proof of length $\Omega\left(\frac{n^{1-\alpha}}{q}\right)$. This is slightly weaker than Theorem~\ref{thm:mainthm} proved in the present paper. Their lower bound is proved by an extension of the communication complexity technique of Brody et.al.~\cite{CCLB} to \emph{Merlin-Arthur communication complexity}. This proof technique is fitting for the MAP setting, but does not apply to partial testing in general. Gur and Rothblum also prove that this trade off is almost optimal for the Tensor Sum properties.

Additionally, Gur and Rothblum show separations between the power of MAPs and that of IPPs and PCPPs. They also show that 2-sidedness may only give a MAP a $\mathrm{polylog}(n)$ factor improvement in proof length and query complexity over a 1-sided algorithm. Their result implies a connection also between 1-sided and 2-sided partial testability, though not one that would preserve O(1)-query partial testability.

\section{Proof plan}

For the proofs of our main result we develop new techniques that are in some ways more flexible than the traditional use of Yao's method for proving property testing lower bounds. We believe that these techniques hold promise for other lower bound situations where using Yao's method seems to hit a wall.

\subsection{General themes}

As with Yao's method, we contrast the behavior of a supposed test when it is run over an input chosen according to some distribution over ``yes'' instances, with its behavior when it is run over an input chosen according to some distribution over ``no'' instances. However, while in the traditional method these two distributions are chosen based only on the property (and should work against all possible algorithms of a given class), here the distributions are in fact {\em tailor made} for the specific analyzed algorithm. Note that special care must be taken in the definition of such an input distribution. It may not depend on the ``real-time'' behavior of the algorithm (i.e.\ it may not adapt itself to the identity of the random queries that the algorithm has made), and is instead constructed based only on the {\em description} of the algorithm.

The second theme is the use of {\em Shannon entropy}. Our goal here is to prove that if $C$ is $C'$-partially testable, then $C'$ cannot be too large. For achieving this we assume that a testing algorithm exists, and then contrast a uniformly random choice of a word in $C'$ with another word chosen from a ``dangerous'' distribution over words far from $C$. The assumption that the test in fact distinguishes the two distributions allows us to show that a uniformly random choice of a word in $C'$ has low entropy, and hence $C'$ must be small. Using entropy instead of direct counting is crucial to using our main method for obtaining a bound against $2$-sided error tests, rather than only $1$-sided error ones.

\subsection {Proving a bound against non-adaptive algorithms}

The bound against non-adaptive algorithms showcases the above general themes. A supposed $C'$-partial test with $q$ queries is in essence a distribution over query sets of size $q$, such that with high probability a query set is chosen that highlights a difference between members of $C'$ and inputs far from being in $C$. Now if the test were additionally $1$-sided, this would translate to forbidden values in the meaningful query sets, which would give a cross product bound on the size of $C'$. For $2$-sided tests we use an analogous feature of entropy, namely that of subadditivity.

To construct a ``dangerous'' distribution over words far from being in $C$, we first take note of the ``heavy'' indexes, which are those bits of the input that are with high probability part of the query subset of the investigated testing algorithm. There will be only few of those, and our distribution over far words would be that of starting with a restriction of a uniformly random word in $C'$ to the set of heavy indexes, and augmenting it with independently and uniformly chosen values to all other input bits. When contrasted with the uniform distribution over all members of $C'$, we obtain that there must be many query sets that show a distinction between the two distributions over the non-heavy indexes with respect to the heavy ones. This means that the values of the non-heavy indexes in each such query set do not behave like a uniformly independent choice, and thus have a corresponding entropy (conditioned on the heavy index bits) that is significantly less than the maximal possible entropy. Having many such query sets in essence means that we can find many such sets that are disjoint outside the heavy indexes, which in turn leads to an entropy bound by virtue of subadditivity (when coupled with general properties of linear codes).

\subsection{Proving a bound against adaptive algorithms}

An adaptive algorithm cannot be described as a distribution over query sets, but rather as a distribution over small decision trees of height $q$ that determine the queries. Therefore low-entropy index sets cannot be readily found (and in fact do not always exist). To deal with this we employ a new technique, that allows us to ``rearrange'' the input in a way that preserves entropy, but now admits disjoint low-entropy sets.

This new construction is a {\em reader}, which in essence is an adaptive algorithm that reads the entire input bit by bit (without repetitions). As this adaptive algorithm always eventually reads the entire input, it defines a bijection between the input to be read and the ``reading stream'', i.e.\ the sequence of values in the order that the reader has read them.

The construction of this reader is fully based on the description of the $q$-query adaptive algorithm that $C'$-partially tests for $C$ (again we assume that such an algorithm exists). In fact we contrast the uniform distribution over members of $C'$ with not one but many possible distributions over inputs far from $C$. At every stage we obtain that, as long as our reader has not yet read a large portion of the input, the adaptive test can provide a decision tree over the yet-unread bits that shows a difference between a uniformly random member of $C'$ (conditioned on the values of the bits already read) and an independently uniform random choice of values for the unread bits. Our reader will be the result of ``concatenating'' such decision trees as long as there are enough unread bits. Thus in the ``reading stream'' we have sets of q consecutive bits, each with low entropy (as it is distinguishable from independently uniform values). When there are not enough unread bits left, we read all remaining bits arbitrarily, and use general properties of codes to bound the entropy on that final chunk.

The method of constructing a reader not only allows us to do away with the exponential penalty usually associated with moving from non-adaptive to adaptive algorithms, but we additionally obtain better bounds for non-adaptive algorithms as well. This is because a reader can do away also with the penalty of moving from the situation of having many low-entropy query sets to having a family of sets disjoint outside the heavy indexes, in essence by constructing the reader for the uniform distribution over $C'$ based on not one but many ``dangerous'' input distributions.

\section{Preliminaries}

Below we introduce the reader to some basic definitions and results regarding entropy and the dual distance of codes. We refer the reader who is interested in a more thorough introduction of entropy to \cite[Chapter 2]{CovThom}.

First, we introduce a standard notion of distance between distributions.

\begin{definition}[Total variation distance]
Let $p$ and $q$ be two distributions over the domain $\mathcal{D}$. The total variation distance between $p$ and $q$ is defined to be $d_{TV}(p,q)=\frac{1}{2}\sum_{i\in \mathcal{D}} |p(i)-q(i)|$.
\end{definition}

We now introduce the notion of the entropy of a random variable, the entropy of a random variable conditioned on another one, and two well-known lemmas.

\begin{definition}[Entropy]
Let $X$ be a random variable over the domain $\mathcal{D}$. The \emph{entropy of $X$} is defined to be $H[X]=-\sum_{i\in\mathcal{D}} \Pr[X=i]\log(\Pr[X=i])$.
\end{definition}

\begin{definition}[Conditional entropy]
Let $X$ and $Y$ be random variables over the domain $\mathcal{D}$. The \emph{entropy of $X$ conditioned on $Y$} is defined to be $H[X\vert Y]=\sum_{y\in\mathcal{D}} \Pr[Y=i]H[X\vert Y=y]$.
\end{definition}

\begin{lemma}[The chain rule]
Assume that $X$ and $Y$ are random variables. The entropy of the combined state determined by both random variables is denoted $H[X,Y]$. This quantity obeys the chain rule $H[X,Y]=H[X\vert Y]+H[Y]$.
\end{lemma}

\begin{lemma}[Subadditivity]
If $X$ and $Y$ are random variables, then $H[X,Y]\leq H[X]+H[Y]$.
\end{lemma}

The total variation distance is not a natural fit to the context of entropy. A more fitting notion of distance between distributions is divergence.

\begin{definition}[Divergence]
Let $p$ and $q$ be two distributions over $\mathcal{D}$. The \emph{divergence of $q$ from $p$} is defined to be $D(p\Vert q)=\sum_{i\in\mathcal{D}} p(i)\log\left(\frac{p(i)}{q(i)}\right)$.
\end{definition}

Fortunately, divergence and total variation distance are related via Pinsker's inequality.

\begin{lemma}[Pinsker's inequality]
Assume that $p$ and $q$ are two distributions over the domain $\mathcal{D}$. The total variation distance between $p$ and $q$ is related to the divergence of $q$ from $p$ by the inequality $\sqrt{\frac{1}{2} D(p||q)}\geq d_{TV}(p,q)$.
\end{lemma}

We will actually be using a simpler corollary of it.

\begin{lemma}[Corollary of Pinsker's inequality]\label{cor:pinsker}
Assume that $X$ is a random variable distributed according to $p$ over $\mathcal{D}$, and denote the uniform distribution over $\mathcal{D}$ by $p_u$. The entropy of $X$ is related to its total variation distance from the uniform distribution by $H[X]\leq\log(|\mathcal{D}|) - 2(d_{TV}(p,p_u))^2$.
\end{lemma}

\begin{proof}
$$H[X]=-\sum_{i\in\mathcal{D}} \Pr[X=i]\log(\Pr[X=i])$$
$$=-\sum_{i\in\mathcal{D}} \Pr[X=i]\log(\Pr[X=i]\cdot\frac{1}{|\mathcal{D}|}\cdot |\mathcal{D}|)$$
$$=-\sum_{i\in\mathcal{D}} \Pr[X=i]\log(\frac{1}{|\mathcal{D}|})-\sum_{i\in\mathcal{D}} \Pr[X=i]\log(\Pr[X=i]\cdot |\mathcal{D}|)$$
$$=\log(|\mathcal{D}|)-D(p\Vert p_u)\leq \log(|\mathcal{D}|) - 2(d_{TV}(p,p_u))^2$$
Where the last step follows from Pinsker's inequality.
\end{proof}

Let $x\in \{0,1\}^n$ and $J\subseteq [n]$. We use $x[J]$ to denote the restriction of $x$ to the indices in $J$. That is, the vector $<x_j>_{j\in J}$. When $C\subseteq\{0,1\}^n$ we use $C[J]=\{x[J]\vert x\in C\}$.

Let $C\subseteq\{0,1\}^n$. We denote by $U(C)$ the uniform distribution over $C$. In accordance with the notation above, when $X\sim U(C)$, $X[J]$ denotes the random variable obtained by drawing uniformly from $C$ and then restricting to the indices in $J$. As a shorthand we use $U(C)[J]$ for the distribution of $X[J]$. We use $U_J(C)$ to denote the result of first drawing a vector $x$ according to $U(C)$, and then replacing $x\left[[n]\setminus J\right]$ with a uniformly random vector in $\{0,1\}^{n-|J|}$. In particular, in many cases we will take $C$ to be a singleton, in which case we drop the curly braces and denote this probability distribution by $U_J(x)$.

We will make inherent use of the following result, which can be found e.g.\ in \cite[Chapter 1, Theorem 10]{MacwSlo}.

\begin{lemma}\label{thm:unifsmall}
If $J\subseteq [n]$ is such that $|J|<\Gamma$ and $X\sim U(C)$, then $X[J]$ is distributed uniformly over $\{0,1\}^{|J|}$.
\end{lemma}

We will also need the fact that a mostly random input is far from a code with high probability.

\begin{lemma}\label{lem:unibfar}
Let $C\subseteq\{0,1\}^n$ such that $|C|\leq 2^{\frac{1}{64}n}$, $\epsilon<1/8$, and let $J\subseteq [n]$ be such that $|J| \leq n/2$. $X\sim U_J(C)$ is $\epsilon$-far from $C$ with probability $1-o(1)$. Furthermore, this is still true when conditioned on any value of $X[J]$.
\end{lemma}

\begin{proof}
By Chernoff bounds, the probability that a random element $X\sim U_J(C)$ will agree with $c\in C$ in more than $(1-\epsilon) n$ coordinates is at most $\exp\bigg(-n(1/4-\epsilon)^2\bigg)$. Taking the union bound over all $c\in C$ gives us $|C|\cdot\exp\bigg(-n(1/4-\epsilon)^2\bigg)=o(1)$. Since this calculation assumes that $X[J]$ always agrees with $c[J]$, it holds when conditioned on any value of $X[J]$.
\end{proof}

Finally, we will also need to use Lemma \ref{thm:unifsmall} to help us calculate the entropy of uniform random variables in codes.

\begin{lemma}\label{lem:restofbits}
Let $C$ be a code with dual distance $\Gamma$, $J\subseteq [n]$ such that $|J|\leq \Gamma$, $C'\subseteq C$ and $X\sim U(C')$. Then $H[X\vert X[J]]\leq \log|C| - |J|$. Furthermore, this is true when conditioned on any particular value of $X[J]$.
\end{lemma}

\begin{proof}
We can partition $C$ according to the values of the bits in $J$:
$$C=\bigcup_{z\in \{0,1\}^{|J|}}\{c\in C\vert c[J]=z\}$$
By Lemma \ref{thm:unifsmall}, all sets on the right hand side are of size $2^{-|J|}|C|$. Obviously, for all $z\in \{0,1\}^{|J|}$, we have $\{c'\in C'\vert c'[J]=z\}\subseteq \{c\in C\vert c[J]=z\}$, simply because $C'\subseteq C$. Thus for every $x\in C'[J]$, we have that
$$H[X\vert X[J]=x]\leq\log|\{c'\in C'\vert c'[J]=z\}|\leq\log|\{c\in C\vert c[J]=z\}|.$$

This completes the ``furthermore'' part of the lemma. To obtain the non-conditioned version, note that by the definition of conditional entropy, $$H[X\vert X[J]]=\mathrm{E}_{x\sim U(C')[J]}H[X\vert X[J]=x]\leq \log\left(2^{-|J|}|C|\right)=\log|C|-|J|.$$
\end{proof}

We note (and use throughout) that trivially $H[X\vert X[J]]=H[X[\{1,\ldots,n\}\setminus J]\vert X[J]]$.

\section{Nonadaptive lower bound}\label{sec:nonadap}

In this section we prove Theorem \ref{thm:mainthm} for the case of a non-adaptive tester and with slightly worse quantitative bounds. For the rest of this section, set $C\subset \{0,1\}^n$ to be a code with dual distance $\Gamma$ and $|C|\leq 2^{\frac{1}{64}n}$. Set $\epsilon<1/8$ and assume that $C$ is $C'$-partially testable for $C'\subseteq C$ with $q$ non-adaptive queries.

Next we define a non-adaptive tester for a property. This definition is consistent with the standard one.
\begin{definition}[Non-adaptive property tester]
A non-adaptive \emph{$\epsilon$-tester} for a code $C\subseteq\{0,1\}^n$ with query complexity $q(\epsilon,n)$ is defined by a collection of query sets $\{Q_i\}_{i\in I}$ of size $q$ together with a predicate $\pi_i$ for each query set and a distribution $\mu$ over $I$ which satisfies:
\begin{itemize}
\item If $x\in C$, then with probability at least $2/3$ an $i\in I$ is picked such that $\pi_i(x[Q_i])=1$.
\item If $d(x,C)>\epsilon$, then with probability at least $2/3$ an $i\in I$ is picked such that $\pi_i(x[Q_i])=0$.
\end{itemize}
For a $C'$-partial tester the first item must hold only for $x\in C'$.
\end{definition}

Set a non-adaptive tester for $C'$, and let $\{Q_i\}_{i\in I}$ be its query sets.

We will be interested only in those query sets which are useful for telling a random element in $C'$ from a mostly random element in $\{0,1\}^n$.

\begin{definition}[$J$-Discerning query set]
Let $J\subseteq [n]$ be such that $|J| \leq n/2$. A query set $Q_i$ is a \emph{$J$-discerning} set if $d_{TV}(U(C')[Q_i],U_J(C')[Q_i])\geq 1/8$.
\end{definition}

Next we prove that a tester must have a lot of such good query sets.

\begin{lemma}\label{lem:goodq}
Set $J\subseteq [n]$ such that $|J| \leq n/2$. With probability at least $1/9$ the query set $Q_i$ picked by the tester is a $J$-discerning set.
\end{lemma}

\begin{proof}
Assume the contrary, that is, that with probability greater than $8/9$ the query set $Q_i$ picked by the tester is such that $d_{TV}(U(C')[Q_i],U_J(C')[Q_i])< 1/8$.

Thus for every such $Q_i$,  $$|\Pr_{U(C')[Q_i]}[\text{tester accepts}]-\Pr_{U_J(C')[Q_i]}[\text{tester accepts}]|<1/8.$$ For the case where the query set picked is not discerning, which occurs with probability smaller than $1/9$, we have no bound (better than $1$) on the difference in probability.

Overall, over the randomness of the tester,
$$|\Pr_{U(C')}[\text{tester accepts}]-\Pr_{U_J(C')}[\text{tester accepts}]|<8/9 \cdot 1/8 + 1/9 = 2/9.$$
But by the correctness of the tester and Lemma \ref{lem:unibfar}, we arrive at $\Pr_{U(C')}[\text{tester accepts}]\geq 2/3$ and $\Pr_{U_J(C')}[\text{tester accepts}]\leq 1/3$, a contradiction.
\end{proof}

We will later want to construct a collection of $J$-discerning sets disjoint outside of a small fixed portion of the input. Towards this end we prove that $J$-discerning sets show difference between an element in $C'$ and a mostly random element in $\{0,1\}^n$ even when we only look outside of $J$.

\begin{lemma}\label{lem:setbadlowent}
Assume that $Q_i$ is a $J$-discerning set, draw $Z\sim U(C')[J]$ and then draw $X\sim U(C')[Q_i]$ conditioned on $X[J]=Z$. With probability at least $1/15$, the distribution of $X[Q_i\setminus J]$ is $1/16$-far from $U(\{0,1\}^{|Q_i\setminus J|})$.
\end{lemma}

\begin{proof}
First note that the distance between $U(C')[Q_i]$ and $U_J(C')[Q_i]$ is the expectation over $Z$ of the distance of $X[Q_i\setminus J]$ from $U(\{0,1\}^{|Q_i\setminus J|})$, conditioned on $X[J]=Z$. By definition, that is at least $1/8$. By simple probability bounds, with probability at least $1/15$, $Z$ is such that the distance of $X[Q_i\setminus J]$ from $U(\{0,1\}^{|Q_i\setminus J|})$ conditioned on $X[J]=Z$ is at least $1/16$.
\end{proof}

However, total variation distance is not very handy for counting. We now use Lemma \ref{cor:pinsker} to transform our total variation bounds into ``entropy loss'' bounds.

\begin{lemma}\label{lem:discent}
If $Q_i$ is a $J$-discerning set and $X\sim U(C')[Q_i]$, then $H[X[Q_i\setminus J]\vert X[J]]\leq |Q_i\setminus J|-0.0005$.
\end{lemma}

\begin{proof}
Let $L\subseteq \{0,1\}^{|J|}$ be the set of values $z\in \{0,1\}^{|J|}$ such that when drawing $X\sim U(C')[Q_i]$ conditioned on $X[J]=z$, the distribution of $X[Q_i\setminus J]$ is $1/16$-far from $U(\{0,1\}^{|Q_i\setminus J|})$.

Since the entropy is non-negative, we can upper bound
$$H[X[Q_i\setminus J]\vert X[J]]\leq\sum_{z\in L} \Pr_{Z\sim U(C')[J]}[Z=z] H[[Q_i\setminus J]\vert X[J]=z]+\sum_{z\in\{0,1\}^{J}\setminus L} \Pr_{Z\sim U(C')[J]}[Z=z] |Q_i\setminus J|.$$
To treat the first summand on the right hand side, we invoke Lemma \ref{cor:pinsker} to obtain
$$H[[Q_i\setminus J]\vert X[J]=z]\leq |Q_i\setminus J|-0.007.$$
Overall we get
$$\sum_{z\in L} \Pr_{Z\sim U(C')[J]}[Z=z] H[[Q_i\setminus J]\vert X[J]=z]+\sum_{z\in\{0,1\}^{J}\setminus L} \Pr_{Z\sim U(C')[J]}[Z=z] |Q_i\setminus J| \leq |Q_i\setminus J| - 0.0005.$$
\end{proof}

Next we would try to cover the indices in $[n]$ with as many discerning sets as possible. To later sum up the entropy loss, we need these sets to be disjoint outside a not-too-big set. We determine this set of ``bad'' indices as the set of bits read (non-adaptively) by the tester with the highest probability.

\begin{definition}
Define $B=\{k\in [n]\vert \Pr_{Q\sim\mu}[k\in Q] \geq  \frac{2q}{\Gamma}\}$.
\end{definition}

\begin{observation}
$|B|\leq \Gamma/2\leq n/2$. Therefore Lemma \ref{lem:unibfar} holds with $I=B$.
\end{observation}

Now we can prove that we can find many $B$-discerning sets which are disjoint outside of $B$.

\begin{lemma}\label{lem:disccover}
There exists a set $I_D$ such that:
\begin{itemize}
\item For all $i\in I_D$, $Q_i$ is a $B$-discerning set
\item For all $i,j\in I_D$, $Q_i\setminus B$ and $Q_j\setminus B$ are disjoint
\end{itemize}
$D=\cup_{i\in I_D}(Q_i\setminus B)$ satisfies $\Gamma/2\geq |D|\geq\frac{\Gamma}{18q^2}$. Additionally, $|I_D|\geq \frac{\Gamma}{18q^3}$.
\end{lemma}

\begin{proof}
We construct the set $I_D$ greedily. Suppose that we have discerning sets covering $k$ bits that are disjoint outside of $B$. Choose a set randomly using the tester's distribution conditioned on it being $B$-discerning. This increases the probability of every query set, and every bit to be in a query set, by at most $9$. By the definition of $B$, if we choose a query set randomly using the tester's distribution, the probability that it intersects our already covered bits outside of $B$ is at most $9\frac{2q^2}{\Gamma}k$. As long as this number is smaller than $1$, such a set exists. Therefore, as long as $k < \frac{\Gamma}{18q^2}$ we have a set to add, leading to the bound. To get the upper bound on $|D|$ we can just stop the process before $D$ gets too big.
\end{proof}

Finally, we are ready to calculate the entropy of a uniformly random codeword from $C'$. We use the chain rule to split this into calculating the entropy of the bits in $B$, the entropy of the bits in $D$ conditioned on the bits of $B$, and the entropy of everything else conditioned on the bits in $D\cup B$.

\begin{lemma}
If $X\sim U(C')$, then $H[X]\leq \log|C| - 0.0005\frac{\Gamma}{18q^3}$
\end{lemma}

\begin{proof}
First, by the chain rule for entropy and the fact that $D\setminus B=D$,

$$H[X]=H[X\vert X[D\cup B]]+H[X[D]\vert X[B]]+H[X[B]]$$

We proceed by bounding each element in the sum. First, trivially:

$$H[X[B]]\leq |B|$$

Next, invoke Lemma \ref{lem:restofbits} for $D\cup B$, since $|D\cup B|\leq \Gamma$. This gives us:

$$H[X\vert X[D\cup B]]\leq \log|C| - |D\cup B|$$

Now, recall that $\cup_{i\in I_D} (Q_i\setminus B)=D$. Since these sets are disjoint outside of $B$, we employ subadditivity to get:
$$H[X[D\setminus B]\vert X[B]]\leq \sum_{i\in I_D} H[X[Q_i\setminus B]\vert X[B]]$$
Now, since these are all $B$-discerning sets, by Lemma \ref{lem:discent} we know that for all $i\in I_D$ we have that $H[X[Q_i\setminus B]\vert X[B]]\leq |Q_i\setminus B| - 0.0005$. By Lemma \ref{lem:disccover} we know that $|I_D|\geq \frac{\Gamma}{18q^3}$. Summing up we get:
$$\sum_{i\in I_D} H[X[Q_i\setminus B]\vert X[B]]\leq |D| - 0.0005|I_D|$$
$$\leq |D| - 0.0005\frac{\Gamma}{18q^3}$$

That is,

$$H[X[D]\vert X[B]]\leq |D| - 0.0005\frac{\Gamma}{18q^3}$$

Summing everything up we get the statement of the lemma.
\end{proof}

From this it follows that:

\begin{theorem}[Weak form of the main theorem]
Let $C'\subseteq C$, if $C$ is $C'$-partially testable with $q$ non-adaptive queries, then $$|C'|=2^{H[X]}\leq |C|2^{- 0.0005\frac{\Gamma}{18q^3}}.$$
\end{theorem}

\section{Adaptive lower bound}\label{sec:adap}

In this section we prove Theorem \ref{thm:mainthm} in its full generality. We start by introducing the mechanism of a \emph{reader}, which allows us to separate the adaptivity and randomness of the algorithm.

\begin{definition}[Reader]
A \emph{$k$-reader} $r$ is a sequence $r_0,r_1,\ldots,r_{k-1}$, where $r_i:\{0,1\}^i\to \{1,\ldots,n\}$ satisfy for all $i<j$ and $y\in\{0,1\}^j$ that $r_i(y[\{1,\ldots,i\}])\neq r_j(y)$.
\end{definition}

Given an input $x\in\{0,1\} ^n$, the reader defines a sequence of its bits. This is the \emph{reading} of $x$.

\begin{definition}[Reading]
Given $x\in\{0,1\} ^n$ and a  $k$-reader $r$, the \emph{reading} $R_{r(x)}$ of $x$ according to $r$ is a sequence $y_1,\ldots,y_k$ defined inductively by $y_{i+1}=x_{r_i(y_1,\ldots,y_i)}$. We define $r_i(x)$ to be $r_i(y_1,\ldots,y_i)$. The set of \emph{unread bits} $U_{r(x)}$ is the subset of $\{1,\ldots,n\}$ that did not appear as values of $r_1,\ldots, r_k$ in the reading.
\end{definition}

We can now define an adaptive tester as a distribution over readers and decision predicates.

\begin{definition}[Adaptive tester]
An \emph{adaptive} $\epsilon$-tester for a code $C\subseteq\{0,1\}^n$ with query complexity $q=q(\epsilon,n)$ is defined by a collection of $q$-readers $\{r^i\}_{i\in I}$ together with predicates $\pi_i$ for each reader, and a distribution $\mu$ over $I$ which satisfies:
\begin{itemize}
\item For all $x\in C$, $\Pr_{i\sim \mu} \left[\pi_i(R_{r^i(x)})=1\right]\geq 2/3$.
\item For all $x\in\{0,1\}^n$ such that $d(x,C)>\epsilon$, $\Pr_{i\sim \mu} \left[\pi_i(R_{r^i(x)})=0\right]\geq 2/3$.
\end{itemize}
\end{definition}

Part of the usefulness of readers is that if we can construct a reader that reads the entire input, then reading the property $C'$ through it preserves its size.

\begin{observation}\label{obs:readallbij}
If $r$ is an $n$-reader, then the function mapping every $x\in\{0,1\}^n$ to its reading $R_{r(x)}$ is a bijection.
\end{observation}

\begin{proof}
Suppose that $x'\neq x$, and let $i\in \{1,\ldots,n\}$ be the least index such that $x_{r_i(x)}\neq x'_{r_i(x)}$. Such an $i$ must exist since $r$ reads all bits, and $x'\neq x$. Note that $r_i(x)=r_i(x')$, since it is the first bit read to be different (and thus $y_1,\ldots,y_i=y'_1,\ldots,y'_i)$. Thus $x_{r_i(x)}\neq x'_{r_i(x')}$ and therefore $R_{r(x)}\neq R_{r(x')}$.
\end{proof}

In light of the above, we will construct an $n$-reader and bound the size of $C'$ when permuted by its reading. However, while the end product of the construction is an $n$-reader, the intermediate steps might not be $k$-readers for any $k$. Thus we need to introduce a more general notion.

\begin{definition}[Generalized reader]
A \emph{generalized reader} $r$ is a sequence $r_0,r_1,\ldots,r_{n-1}$ where $r_i:\{0,1\}^i\to \{1,\ldots,n\}\cup\{\star\}$ satisfy for all $i<j$ and $y\in\{0,1\}^j$ one of the following
\begin{itemize}
\item $r_i(y[\{1,\ldots,i\}])\in \{1,\ldots,n\}\setminus r_j(y)$
\item $r_i(y[\{1,\ldots,i\}]) = r_j(y) = \star$
\end{itemize} 
Given a generalized reader $r$, a {\em terminal sequence} in it is $y\in\{0,1\}^i$ such that $r_i(y_1,\ldots,y_i)=\star$, while $r_{i-1}(y_1,\ldots,y_{i-1})\neq\star$ or $i=0$.
\end{definition}

If we fix a certain $x\in\{0,1\} ^n$, a generalized reader defines a sequence of non-repeating indices that at some point may degenerate to a constant sequence of $\star$. Note that every $k$-reader naturally defines a generalized reader by setting all undefined functions to map everything to $\star$.

It is useful to think of a (possibly generalized) reader as a decision tree. With a generalized reader, we will often want to continue the branches of the tree with another reader. This operation is called \emph{grafting}. We start with the notion of a \emph{$0$-branch} and a \emph{$1$-branch}.

\begin{definition}[$0$-branch, $1$-branch]
Let $r$ be a (possibly generalized) reader. The \emph{$0$-branch of $r$} is the reader $r'$ defined by $r'_i(y_1,\ldots,y_i)=r_{i+1}(0,y_1,\ldots,y_i)$. Similarly, the $1$-branch of $r$ is the reader $r''$ defined by $r''_i(y_1,\ldots,y_i)=r_{i+1}(1,y_1,\ldots,y_i)$.
\end{definition}

We can now define grafting, and will do so recursively. Informally, grafting a reader $t$ onto $r$ at $y$ means that at every $\star$ in the decision tree of $r$ that can be reached after reading $y$, we continue the reading according to $t$.

\begin{definition}[Grafting]
Let $r$ and $t$ be generalized readers and $x\in\{0,1\}^i$ be a terminal sequence in $r$. The \emph{grafting of $t$ onto $r$ on the branch $y$} is a new reader $r^{t,y}$ defined as follows.
\begin{itemize}
\item If $t_0\in \{r_0(y_1,\ldots,y_i),\ldots,r_{i-1}(y_1,\ldots,y_i)\}$, graft the $y_{t_0}$-branch of $t$ onto $r$ at $y_1,\ldots,y_i$.
\item If $t_0\notin \{r_0(y_1,\ldots,y_i),\ldots,r_{i-1}(y_1,\ldots,y_i)\}$, set $r_i(y_1,\ldots,y_i)=t_0$, call the new reader $r'$, and graft the $0$-branch of $t$ onto $r'$ at $y_0,\ldots,y_i,0$ and the $1$-branch of $t$ onto $t$ at $y_0,\ldots,y_i,1$.
\end{itemize}
Repeat the above recursively, with the base case being the grafting of an identically $\star$ reader onto $r$ by not changing anything.
\end{definition}

Note that the grafting of a generalized reader onto another results in a generalized reader. Note that it is also possible that $r^{t,y}=r$ when all bits that $t$ may read were already read by $r$ according as $y$.

To introduce the notion of a reader that discerns a random input from an input from $C'$, we will first need to formulate a notion of executing a reader, which is inherently adaptive, on a partly random input.

\begin{definition}[$J$-Simulation of a reader]
Let $r$ be a $q$-reader, $J\subseteq [n]$ and $y\in\{0,1\} ^{|J|}$. The \emph{$J$-simulation of $r$ on $y$} is the distribution $S(r,y,J)$ over $\{0,1\}^q$ defined to be $R_{r(x)}$ where $x[J]=y[J]$, and all bits of $x$ outside of $J$ are picked independently and uniformly at random from $\{0,1\}$.
\end{definition}

We now introduce the notion of a reader that discerns a random input from an input from $C'$.

\begin{definition}[$J$-Discerning reader]
Let $r$ be a (possibly generalized) reader, $J\subseteq [n]$ and $y\in\{0,1\} ^{|J|}$. Let $x$ be a uniform random variable in $\{c\in C'\vert c[J]=y\}$. We say that $r$ is a \emph{$J$-discerning reader for $y$} if $d_{TV}(R_{r(x)},S(r,y,J))\geq 1/8$.
\end{definition}

Next, we prove that many readers are indeed discerning.

\begin{lemma}\label{lem:goodr}
Set $J\subseteq [n]$ such that $|J| \leq n/2$ and $y\in\{0,1\} ^{|J|}$. With probability at least $1/9$ the $q$-reader $r$ picked by the tester is $J$-discerning for $y$.
\end{lemma}

\begin{proof}
Let $r$ be a reader that is not $J$-discerning for $y$. Let $B\sim U_J(y)$ and $G\sim U(\{c\in C'\vert c[J]=y\})$. Denote by $\pi_r$ the predicate associated with $r$. By our assumption, 
$$|\Pr[\pi_r(R_{r(B)})=1]-\Pr[\pi_r(R_{r(G)})=1]|<1/8.$$

Now assume that with probability greater than $8/9$, the $q$-reader picked is not $J$-discerning for $y$. Now consider the difference in acceptance probability when drawing a reader according to $\mu$.

$$|\Pr_{r\sim\mu}[\pi_r(R_{r(B)})=1]-\Pr_{r\sim\mu}[\pi_r(R_{r(G)})=1]|< 8/9\cdot 1/8+1/9=2/9.$$

But by Lemma \ref{lem:unibfar} and the correctness of the tester, $\Pr_{r\sim\mu}[\pi_r(R_{r(B)})=1]\leq 1/3$, and by the correctness of the tester $\Pr_{r\sim\mu}[\pi_r(R_{r(G)})=1]\geq 2/3$, a contradiction.
\end{proof}

A common operation will be to graft a discerning reader with additional arbitrary bits. This does not cause a discerning reader to stop being one.

\begin{definition}
Let $r$ and $s$ be generalized readers. We say that \emph{$r$ contains $s$} if for every $x\in\{0,1\}^n$, the sequence of non-$\star$ elements in $R_{s(x)}$ is a prefix of $R_{r(x)}$.
\end{definition}

Note that in particular, whenever we graft $s$ onto $r$ along some branch, we obtain a reader which contains $r$.

\begin{lemma}\label{lem:discpreserv}
Let $r$ and $s$ be generalized readers such that $r$ contains $s$. Let $J\subseteq [n]$ and $y\in\{0,1\} ^{|J|}$. If $s$ is a $J$-discerning reader for $y$, then so is $r$.
\end{lemma}

\begin{proof}
Let $B\sim U_J(y)$ and $G\sim U(\{c\in C'\vert c[J]=y\})$. Consider $R_{r(B)}$. Its outcomes can be partitioned according to their $R_{s(B)}$ prefixes. Thus the probability of every event defined by values of $R_{r(B)}$ can be written as a weighted sum of the probabilities of events defined by values of $R_{s(B)}$. The same is true for $R_{r(G)}$ and $R_{s(G)}$. Therefore $d_{TV}(R_{r(x)},S(r,y,J))\geq d_{TV}(R_{s(x)},S(s,y,J))$.
\end{proof}

To prove that a uniform choice in $C'$ does not have high entropy we graft discerning readers one onto the other. We will want to make sure that all the branches of the decision tree are of the same height throughout the grafting, and thus we define the notion of a \emph{padded grafting}.

\begin{definition}[$q$-Padded grafting]
Let $r$ be a generalized reader, $t$ be a $q$-reader and $y\in\{0,1\}^i$ be a terminal sequence in $r$. The \emph{$q$-padded grafting of $t$ onto $r$ on the branch $y$} is defined by the following process. First, let $r'$ be the grafting of $t$ onto $r$ at the branch $y$. Now perform the following repeatedly: Let $z_1,\ldots,z_j$ with $j<q$ be such that $r'_{i+j-1}(y_1,\ldots,y_i,z_1,\ldots,z_{j-1})\neq\star$ but $r'_{i+j}(y_1,\ldots,y_i,z_1,\ldots,z_j)=\star$ (or $j=0$ and $r'_i(y_1,\ldots,y_i)=\star$). Let $k$ be an arbitrary index not in $\{r'_{0},\ldots,r'_{i+j-1}(y_1,\ldots,y_i,z_1,\ldots,z_{j-1})\}$, and redefine $r'_{i+j}(y_1,\ldots,y_i,z_1,\ldots,z_j)=k$. Repeat this process as long as such $z_1,\ldots,z_j$ with $j<q$ exist.
\end{definition}

The above is basically grafting additional arbitrary reads, so that the end-result will always read exactly $q$ bits after reading the sequence $y_1,\ldots,y_i$. The next observation together with Lemma \ref{lem:discpreserv} implies that $q$-padded grafting of a $J$-discerning reader is equivalent to a grafting of some other $J$-discerning reader.

\begin{observation}
Let $r$ be a generalized reader, $t$ a $q$-reader and $y\in\{0,1\}^i$ a terminal sequence in $r$. There exists a reader $s$ containing $t$ such that the $q$-padded grafting of $t$ onto $r$ at $y$ is equivalent to the grafting of $s$ onto $r$ at $y$.
\end{observation}

Now we can finally prove the main lemma, by performing repeated $q$-padded grafting of discerning readers one onto another.

\begin{lemma}
If $X\sim U(C')$, where $C$ is $C'$-partially testable with $q$ queries, then $H[X]\leq\log|C|-\lfloor\frac{1}{32}\Gamma/q\rfloor$
\end{lemma}

\begin{proof}
Let us construct an $n$-reader and consider the entropy of $C'$ when permuted by this reader.

Start with a $0$-reader $r^0$. Let $s$ be a $\emptyset$-discerning $q$-reader for the empty word, which must exist since the adaptive tester must pick one with positive probability. Set $r^1$ to be the grafting of $s$ onto $r^0$ on the branch of the empty word.

Assume that we have constructed the $jq$-reader $r^j$. If $jq\geq \Gamma$, graft a reader that reads all remaining bits arbitrarily onto $r^j$ on all branches. Else, perform the following for all branches $y\in\{0,1\}^{jq}$ to obtain $r^{j+1}$ (noting that they are all terminal sequences in $r^j$):

\begin{itemize}
\item If there is no member of $C'$ with the reading $R_{r^j(y)}$, perform a $q$-padded grafting of an arbitrary $q$-reader onto $r^j$ at the branch $y$,
\item If such a member exists, let $s$ be a $\{r_1^j(y),r_2^j(y),\ldots, r_{jq}^j(y)\}$-discerning reader for $y$. Perform a $q$-padded grafting of $s$ onto $r^j$ at the branch $y$.
\end{itemize}

Now let $r$ be the resulting $n$-reader, let $r_{R(C')}$ be the image of $C'$ under the reading of $r$, and let $X\sim U(r_{R(C')})$. By Observation \ref{obs:readallbij}, the distribution of $X$ is the same as starting with a uniformly random member of $C'$ and then taking its reading according to $r$. By the chain rule $H[X]=H[X[\{1,\ldots,\Gamma\}]]+H[X\vert X[\{1,\ldots,\Gamma\}]]$. 

Note that in the case of a word from $C'$, the maximal $j$ in the construction is equal to $\Gamma/q$. By the chain rule we may write 
$$H[X[\{1,\ldots,\Gamma\}]]=\sum_{i=1}^{\Gamma/q} H[X[\{(i-1)q+1,\ldots,iq-1\}]\vert X[\{1,\ldots,(i-1)q-1\}]]$$ and since each sequence of $q$ bits is from the grafting of a reader which is discerning with respect to all the previous ones, we may apply Lemma \ref{cor:pinsker} to obtain
$$H[X[\{1,\ldots,\Gamma\}]]=\sum_{i=1}^{\Gamma/q} H[X[\{(i-1)q+1,\ldots,iq-1\}]\vert X[\{1,\ldots,(i-1)q-1\}]]$$ $$\leq\sum_{i=1}^{\Gamma/q}\left( q-\frac{1}{32}\right)\leq \Gamma- \Gamma/q\cdot \frac{1}{32}$$.

By Lemma \ref{lem:restofbits}, $H[X\vert X[\{1,\ldots,\Gamma\}]]\leq \log|C|-\Gamma$, so by summing it all up we get $H[X]\leq \log|C| - \Gamma/q\cdot \frac{1}{32}$.

\end{proof}

This gives us Theorem \ref{thm:mainthm} in its full generality, as it implies that $|C'|=2^{H[X]}\leq2^{-\Gamma/32q}\cdot|C|$.

\bibliographystyle{plain}
\bibliography{advice}

\end{document}